\documentclass[11pt]{amsart}
\usepackage{amsmath,amssymb,latexsym}
\usepackage{multicol}
\theoremstyle{definition}
\newtheorem{theorem}{Theorem}

\newtheorem{proposition}[theorem]{Proposition}
\newtheorem{lemma}[theorem]{Lemma}
\newtheorem{definition}[theorem]{Definition}
\newtheorem{example}[theorem]{Example}

\newtheorem{remark}[theorem]{Remark}

\newcommand{\numberset}{\mathbb}

\newcommand{\N}{\numberset{N}}

\newcommand{\F}{\numberset{F}}

\usepackage{hyperref}

 \usepackage{mathptmx}

\usepackage{amsaddr}

\usepackage[margin=3.2cm]{geometry}

\begin{document}
\title[]{Classification of binary systematic codes of small defect}

\author{Alberto Ravagnani$^1$}

\address{Department of Mathematics, University of Neuch\^{a}tel\\
Rue Emile-Argand 11, CH-2000 Neuch\^{a}tel (Switzerland)}

\email{alberto.ravagnani@unine.ch$^1$}

\maketitle

\begin{abstract}
 In this paper non-trivial non-linear binary
systematic AMDS codes are classified in terms 
of their  weight distributions, employing only elementary techniques. In
particular, we show that 
their length and minimum distance completely determine the weight distribution.
\end{abstract}

\section{Introduction} \label{intr}

Let $q$ be a prime power and let $\F_q$ denote the finite field with $q$
elements.
 A (non-linear) \textbf{code} of lenght $n \in \N_{\ge 1}$ over the field $\F_q$
is  a subset 
$C \subseteq \F_q^n$ with at least two elements. 
We omit the adjective \textit{non-linear} for the rest of the paper.  A code $C
\subseteq \F_q^n$ is said to be \textbf{linear} if it is a
vector subspace of $\F_q^n$. Define the  \textbf{Hamming distance} on
$\F_q^n$ by
$d:\F_q^n \times \F_q^n \to \N$ with $$d(v,w):= |\{1 \le i \le n : v_i \neq
w_i\}|,$$ where $v=(v_1,...,v_n)$ and $w=(w_1,...,w_n)$.
The \textbf{minimum distance} $d(C)$ of a code $C \subseteq \F_q^n$ is
the
integer
$$d(C):= \min \{ d(v,w) : v,w \in C, v \neq w \}.$$
 A code of lenght $n$, $|C|$ codewords and minimum distance $d$ is said to be 
of \textbf{parameters} $[n,|C|,d]$.
The \textbf{weight} of a vector $v \in \F_q^n$ 
is the integer
$\mbox{wt}(v):=|\{ 1 \le i \le n : v_i \neq 0\}|$, where $v=(v_1,...,v_n)$. 
Let $C \subseteq \F_q^n$ be a code
containing zero. For any $i \in \N$ such that  $0
\le i \le n$ we denote by
$W_i(C)$ the number of codewords in $C$ having weight exactly $i$. The
collection
 $\{W_i(C) \}_{0 \le i \le n}$ is said to be the \textbf{weight distribution} of
$C$, and the $W_i(C)$'s are called \textbf{weights}.   The following bound is
well-known
(\cite{Singleton}, Theorem 1).

\begin{proposition}[Singleton bound]\label{sb}
 Let $C \subseteq \F_q^n$ be a code of minimum distance $d$. Then $|C| \le
q^{n-d+1}$.
\end{proposition}

\begin{definition}A code $C \subseteq \F_q^n$ which attains the
Singleton bound is said to be an \textbf{MDS} codes 
(MDS stands for \textit{Maximum Distance Separable}).
  A code $C \subseteq \F_q^n$ with minimum distance $d$ and cardinality
$q^{n-d}$ 
 is said to be an \textbf{AMDS} code (AMDS stands for \textit{Almost} MDS). 
 \end{definition}

 Hence, AMDS codes are codes that almost reach the Singleton bound. They
have been
introduced and studied for the first time in \cite{B}.

The remainder of the paper is organized as follows. Section \ref{Prel} contains
some preliminary results on the parameters of MDS and AMDS binary codes. We
introduce systematic codes in Section \ref{param}, where we also classify binary
systematic AMDS codes with respect to their parameters. In Section
\ref{sectionW} we prove that length and minimum distance completely determine
the weight distribution of such codes, and compute them explicitly.
 
 \section{Preliminaries} \label{Prel}

First, as an application of the well-known Hamming bound (\cite{NTV}, Theorem
1.1.47),
we  prove powerful
restrictions on the
size of MDS and AMDS codes.

\begin{proposition} \label{restrizioni}
 Let $C \subseteq \F_q^n$ be a code of minimum distance $d \ge 3$ and $|C|=q^k$
words.
\begin{enumerate}
 \item If $C$ is an MDS code, then $k \le q-1$.
\item If $C$ is an AMDS code, then $k \le q^2+q-2$.
\end{enumerate}
\end{proposition}
\begin{proof}
Set $s(C):=n-d-k+1$, so that $s(C)=0$ if $C$ is MDS, and $s(C)=1$ if $C$
is  AMDS. Remove from the codewords of $C$ the last $d-3$ components, obtaining
a code, say $D$, of lenght $n-d+3$, minimum distance at least $3$, and
$|C|=q^k$ codewords. Applying the Hamming bound to
 $D$ we get 
$q^{n-d+1-s(C)} \left[ 1+(n-d+3)(q-1)\right] \le q^{n-d+3}$. 
Straightforward computations give the thesis.
\end{proof}

The following classification of binary MDS codes is a well-known result in
coding theory (see for instance \cite{LX}, Problem 5.32). Another proof using
different techniques can be found in \cite{OS}.

\begin{theorem}[Classification of binary MDS codes]\label{thhh} Let $C
\subseteq \F_2^n$ be any MDS code of minimum distance $d$. Then, up to
traslation, $C$ is one of the following MDS codes.
\begin{enumerate}
\item The $n$-times repetition code, with $d=n \ge 3$.
\item  The parity-check code of the code $\F_2^k$, $k=n-1$.
\item The code $\F_2^n$.
\end{enumerate}
\end{theorem}

Proposition \ref{restrizioni} and Theorem \ref{thhh} will be employed in the
following sections to determine the possible parameters and weight distributions
of binary systematic AMDS codes.

\begin{definition}
Codes $C,D$ over a finite field $\F_q$ are said to be \textbf{$P$-equivalent} if
they
have the same parameters, i.e., $[n(C),|C|,d(C)]=[n(D), |D|, d(D)]$. Codes
$C,D$
over $\F_q$ and containing zero are said to be
\textbf{$W$-equivalent} if
they have the same lenght and the same weight distribution.
\end{definition}
\begin{remark}\label{fatti}
 Notice
that if $C$ and $D$ are
linear
$W$-equivalent codes, then they are also $P$-equivalent. For non-linear codes
containing zero, this result is not true in general (see Example
\ref{esempietto}).
\end{remark}

\begin{example} \label{esempietto}
 The binary codes $\{00000, 11001, 00111\}$ and $\{00000, 10011, 11001\}$
contain zero, have the same weight distribution, and different minimum
distances.
\end{example}

\section{Parameters of binary systematic AMDS codes} \label{param}

Here we study binary systematic AMDS
codes providing a classification  in terms of their parameters. 
Let us briefly recall the definition of systematic code.

\begin{definition}\label{sys}
 Let $n$ be a positive integer and $q$ a prime power. A code $C \subseteq
\F_q^n$ 
is said to be \textbf{systematic} if there exists a function
 $\varphi:\F_q^k \to \F_q^{n-k}$ such that:
 \begin{enumerate}
  \item $\varphi(0)=0$,
  \item $C=\{ (v,\varphi(v)) : v \in \F_q^k\}$.
   \end{enumerate}
The function $\varphi$ is a
 \textbf{systematic encoding function}. A code $C$ as in the definition
 has $q^k$ codewords.
  \end{definition}

\begin{remark}
 Notice that condition (1) is not always required in the definition of
systematic code in the literature. On the other hand, up to a translation, we
can always assume that (1) holds without loss of generality.
\end{remark}

Systematic codes turn out
to be very useful in the applications (see \cite{pre} and \cite{pre1} among
others) and
powerful bounds on their parameters have been recently discovered (see e.g. 
\cite{BGS}).

We study first binary systematic AMDS codes of minimum distance
one and two, providing a characterization.

\begin{proposition} \label{AMDSd=1}
 A code $C \subseteq \F_2^n$ of minimum distance $d=1$ and $2^k$ codewords ($k
\ge 1$) is systematic  and AMDS if and only if there
exists a function
 $\psi:\F_2^k \to \F_2$ with the following properties:
\begin{enumerate}
 \item[(a)] $\psi(0)=0$,
\item[(b)] $\psi$ is not the parity-check function,
\item[(c)] $C= \{ (v,\psi(v)) : v \in \F_2^k \}$.
\end{enumerate}
As a consequence, for any $n \ge 2$, there are $2^{2^{n-1}-1}-1$ such codes.
\end{proposition}
\begin{proof}
 Assume that $C$ is systematic and AMDS. Let $\psi:=\varphi$, the encoding
function of Definition \ref{sys}. We clearly have
 $\psi(0)=0$. By contradiction, assume that $\varphi$ is the parity-check
function. Consider two
vectors
 $v,w \in \F_2^k$ such that $d(v,w) = 1$. We have $\mbox{wt}(v) \not
\equiv \mbox{wt}(w)
 \mbox{ (mod 2)}$ and $d((v,\varphi(v)),(w,\varphi(w)))=2$. This proves that
the minimum 
distance of $C$ is two, a contradiction. Now assume that $\psi:\F_2^k \to
\F_2$ satisfies the hypothesis. We
need to prove that the code 
$C:=\{(v,\psi(v)) : v \in \F_2^k\}$ is AMDS. 
By contradiction, assume that $C$ is not AMDS. Since $d(C)$
 trivially satisfies $1 \le d(C) \le 2$, $C$ has $2^k$ elements and length 
$k+1$,
we have that $C$ is an MDS 
code. This contradicts Theorem \ref{thhh}.
\end{proof}

\begin{proposition} \label{AMDSd=2}
The following facts hold.
\begin{enumerate}
 \item A code $C \subseteq \F_2^n$ of minimum distance $d=2$ and $2^k$ elements
($k \ge 2$)
is systematic and AMDS  if and only if there exists a function
$\psi:\F_2^k \to \F_2^2$ such that:
\begin{enumerate}
 \item $\psi(0)=0$,
\item $\psi(v) \neq \psi(w)$ for any $v,w \in \F_2^k$ such that
$d(v,w)=1$,
\item $C= \{ (v,\psi(v)) : v \in \F_2^k \}$.
\end{enumerate}
\item A subset $C\subseteq \F_2^n$ is an AMDS systematic code with two elements
 if and only if it is of the form $\{ 0,(1,v)\}$ with $n \ge
2$, $v \in \F_2^{n-1}$ and $\mbox{wt}(v)=n-2$.
\end{enumerate}
\end{proposition}
\begin{proof}
 If $C$ is systematic and AMDS, let $\psi:=\varphi$, the encoding function
of  Definition \ref{sys}. Properties (a), (b) and (c) are easily checked.
 On the other hand, assume that $\psi:\F_2^k \to \F_2^2$ satisfies
 (a), (b), (c). By (a), the code $C:= \{ (v,\psi(v)) : v \in \F_2^k
\}$ 
is systematic. By (b), the code $C$ has not minimum distance
one. By the Singleton bound, we have $d(C) \in \{ 2,3\}$.
 If $C$ has minimum distance three, then it is an MDS code.
 On the other hand, Theorem \ref{thhh} states that a binary MDS code of 
parameters $[k+2, 2^k,3]$ does not exist ($k \ge 2$). 
Hence $d(C)=2$ and $C$ is AMDS. The last part of the claim is immediate.
\end{proof}

Now we focus on the $P$-classification of binary systematic AMDS
codes with minimum distance at least three.  We start by proving that the
size of any such a code has to be very small.

\begin{lemma}\label{solicasi}
 Let $C$ be a binary systematic AMDS code of minimum distance
$d \ge 3$ and length $n$. 
 Then $k:=n-d \in \{ 1,2,3,4\}$.
Moreover, if $k\in \{2,3,4\}$, then $d \in \{ 3,4\}$.
\end{lemma}
\begin{proof}
 Proposition \ref{restrizioni} gives $k \in \{ 1,2,3,4\}$. If $k \in \{ 2,3
\}$, then
the Plotkin bound (\cite{NTV}, Theorem 1.1.45)
implies $d \in \{ 3,4\}$. If $k=4$, then the same bound gives $d \in \{3,4,5\}$.
The case $k=4$ and $d=5$ is ruled out by the Hamming bound (\cite{NTV}, Theorem
1.1.47).
\end{proof}

\begin{theorem}[$P$-classification]\label{thhh1}
 Let $C$ be a binary systematic AMDS code of length $n$ and minimum distance
$d$. Then 
$(n,d)$ is one of the following pairs:
\begin{enumerate}
\begin{multicols}{2}
\item[(a)] $(n,1)$, with $n \ge 3$,
\item[(b)] $(n,2)$ with $n \ge 4$,
\item[(c)] $(d+1,d)$ with $d \ge 1$,
\item[(d)] $(5,3)$, 
\item[(e)] $(6,4)$,
\item[(f)] $(6,3)$, 
\item[(g)] $(7,4)$, 
\item[(h)] $(7,3)$, 
\item[(i)] $(8,4)$.
\end{multicols}
\end{enumerate}
Moreover, for any such a pair $(n,d)$ there exists a binary systematic AMDS
code of length $n$ and minimum distance $d$. 
\end{theorem}

\begin{proof}
Assume that $(n,d)$ are the length and the minimum distance of a binary systematic AMDS code.
Combining Proposition \ref{AMDSd=1}, Proposition \ref{AMDSd=2} and Lemma \ref{solicasi} we
easily see that $(n,d)$ must be one of the pairs in the list.

We need to show that, for any pair $(n,d)$ in the list, there
exists a binary systematic AMDS code with length $n$ and minimum distance
$d$. Proposition \ref{AMDSd=1} and Proposition \ref{AMDSd=2} produce examples of binary 
systematic AMDS codes
with the parameters of (a), (b) and (c). Notice that, for any $n \ge 5$ and $d
\ge 3$ odd, 
the parity-check code of a
code with length $n$ and minimum distance $d$
has length $n+1$ and minimum distance $d+1$.
As a consequence, it is enough to prove the theorem for the pairs $(5,3)$,
$(6,3)$
and $(7,3)$. For
each tern $[5,2^2,3]$, $[6,2^3,3]$ and $[7,2^4,3]$ we give in Table \ref{3ex}
a generator matrix of a binary linear systematic code having these parameters.
We point out that the code of parameters $[7,2^4,3]$ in the table is
the Hamming code $H(q=2,r=3)$ (see \cite{MS}, page 23). 
\end{proof}

%\vspace{0.2cm}
\begin{table}[h!]
 \renewcommand{\arraystretch}{1.3}
\renewcommand{\tabcolsep}{1mm}
\centering
\caption{Binary systematic codes of parameters $[5,2^2,3]$, 
$[6,2^3,3]$ and $[7,2^4,3]$}
\label{3ex}
\begin{tabular}{|p{1.7cm}||c|c|c|}
$[n,2^k,d]\  \rightarrow$ & $[5,2^2,3]$ & $[6,2^3,3]$ & $[7,2^4,3]$ \\
\hline 
\hline
Generator matrix \ $\rightarrow$& $\begin{bmatrix}
0 & 1 & 0 & 1 & 1 \\
1 & 0 & 1 & 0 & 1
\end{bmatrix}$ & $\begin{bmatrix}
  0  & 0 &  1 &  0 &  1 &  1 \\
0  & 1  & 0 &  1  & 0  & 1 \\
1  & 0 &  0  & 1 &  1 &  0
 \end{bmatrix}$ & $\begin{bmatrix}
1 &  0 &  0  & 0 &  1 &  1 &  0 \\
0  & 1  & 0  & 0 &  1 &  0 &  1  \\
0  & 0 &  1 &  0  & 0  & 1 &  1  \\
0  & 0 &  0  &  1 &  1 &  1 &  1 
 \end{bmatrix}$ \\
\hline

\end{tabular}
\end{table}

\section{Weight distributions of binary systematic AMDS
codes}\label{sectionW}

Here we focus on the $W$-classification of binary systematic AMDS codes of minimum distance
at least three and more than two codewords. By
Lemma \ref{solicasi}, it is enough to study the weight distributions of codes
of parameters $(n,d) \in \{(5,3), (6,4), (6,3), (7,4), (7,3), (8,4)\}$. We will
treat the pairs $(5,3)$ and $(6,4)$ by using a computational approach (the
computations take only a few
seconds on a common 
laptop), and the other cases theoretically. The following lemma is
proved by
exhaustive 
research.

\begin{lemma}\label{AAA}
 The weight distribution of any binary systematic AMDS code of length $n$ and minimum
distance $d$, with $(n,d) \in \{
(5,3), (6,3)\}$, depends only
on $n$ and $d$, and it is given in Table \ref{WclaI}. Moreover, any
such a code is linear.

%\vspace{0.2cm}
\begin{table}[h!]
 \centering
\renewcommand{\arraystretch}{1.3}
\renewcommand{\tabcolsep}{1mm}
\caption{Partial $W$-classification of binary systematic AMDS codes.}
\label{WclaI}
\begin{tabular}{c|p{7.3cm}}
 Values of $(n,d)$ & Non-zero weights of any binary systematic code $C$ with length $n$
and minimum distance $d$ \\
\hline
\hline
$(5,3)$ & $W_0(C)=1$, $W_3(C)=2$, $W_4(C)=1$  \\
\hline
$(6,3)$ & $W_0(C)=1$, $W_3(C)=4$, $W_4(C)=3$ \\
\hline
\end{tabular}
\end{table}
\end{lemma}

Now we focus on the other pairs $(n,d) \in \{ (6,4), (7,4), (7,3), (8,4)\}$
from a theoretical viewpoint. We notice that
exhaustive search does not produce any result in a resonable time on a common
computer when analyzing the cases $(n,d)=(7,3)$ and $(n,d)=(8,4)$. Let us first
recall the definition of weight distribution of a code.

\begin{definition}\label{dd}
 Let $C \subseteq \F_q^n$ be a code over a finite field $\F_q$. For any $i \in
\{0,1,...,n\}$ define the integer $B_i(C)$ by  $|C| \cdot B_i(C):=
|\{ (v,w) \in C^2 : d(v,w)=i\}|$.
The collection ${\{ B_i(C)\}}_{i=0}^{n}$ is called the \textbf{distance
distribution} of $C$.
\end{definition}

\begin{remark} \label{distwe}
 In the notation of Definition \ref{dd}, if $C$ is a linear code then its weight
distribution and its distance distribution agree, i.e., $W_i(C)=B_i(C)$
for any $i \in \{0,1,...,n\}$.
\end{remark}

\begin{lemma}\label{BBB}
A binary systematic AMDS code $C$ of length $n=7$ and minimum distance $d=3$
has the following weight and distance distribution.
\begin{table}[h!]
\centering
\renewcommand{\arraystretch}{1.3}
\renewcommand{\tabcolsep}{1mm}
\begin{tabular}{l|lcl|l}
$W_0(C)=1$ & $W_4(C)=7$ & \ \ \ \ \ \ \ \ \ \ & $B_0(C)=1$ & $B_4(C)=7$ \\
$W_1(C)=0$ & $W_5(C)=0$ & & $B_1(C)=0$ & $B_5(C)=0$ \\
$W_2(C)=0$ & $W_6(C)=0$ & & $B_2(C)=0$ & $B_6(C)=0$ \\
$W_3(C)=7$ & $W_7(C)=1$ & & $B_3(C)=7$ & $B_7(C)=1$
\end{tabular}
\end{table}
\end{lemma}

\begin{proof}
We clearly have  $|C|=2^4=16$, and so the parameters of $C$ 
attain the Hamming bound (\cite{NTV}, Theorem 1.1.47). Such a code is said to be a perfect code 
(see
\cite{MS}, Chapter 6 and \cite{L}, Chapter 11). By \cite{MS}, Theorem 37 at page 182 and the
 following
remark, $C$ has the same weight distribution of the well-known Hamming code
$H(q=2,r=3)$ of parameters $[7,2^4,3]$ (see \cite{MS}, pag. 23). The weight
distribution of this simple linear code is well-known.
\end{proof}

The following result is immediate.
\begin{lemma}\label{tecn}
 Let $n$ be a positive integer and let $v,w \in \F_2^n$. Then
$d(v,w)=\mbox{wt}(v)+\mbox{wt}(v)-2v \cdot w$, where
$v \cdot w= |\{1 \le i \le n : v_i=w_i=1 \}|$. In particular, the integer
$\mbox{wt}(v) -\mbox{wt}(w)$ is odd if and only if  $d(v,w)$ is odd.
\end{lemma}

\begin{theorem}[$W$-classification]
\label{thhh2}
 Any binary systematic AMDS code of minimum distance at least three and
cardinality at least four has exactly one
of the weight distributions listed in Table \ref{Wcla}. Moreover, each of those weight
distribution corresponds to a binary systematic AMDS code.

\begin{table}[h!]
\centering
\renewcommand{\arraystretch}{1.3}
\renewcommand{\tabcolsep}{1mm}
\caption{Weight distribution of any binary systematic AMDS
codes $C$ with minimum distance at least three and cardinality at least
four.}
\label{Wcla}
%\scriptsize{
\begin{tabular}{r||c|c|c|c|c|c|c|}
$[n,2^k,d] \rightarrow$ &$[5,2^2,3]$ & $[6,2^2,4]$ & $[6,2^3,3]$ & $[7,2^3,4]$ &
$[7,2^4,3]$ & $[8,2^4,4]$ \\
\hline
\hline
$W_0(C)$ & $1$ & $1$ & $1$ & $1$ &  $1$& $1$ \\
$W_1(C)$ & $0$ & $0$ & $0$ & $0$ &  $0$& $0$ \\
$W_2(C)$ & $0$ & $0$ & $0$ & $0$ &  $0$& $0$ \\
$W_3(C)$ & $2$ & $0$ & $4$ & $0$ &  $7$& $0$ \\
$W_4(C)$ & $1$ & $3$ & $3$ & $7$ &  $7$& $14$ \\
$W_5(C)$ & $0$ & $0$ & $0$ & $0$ &  $0$& $0$ \\
$W_6(C)$ & -   & $0$ & $0$ & $0$ &  $0$& $0$ \\
$W_7(C)$ & -   & -   &  -  & $0$ &  $1$& $0$ \\
$W_8(C)$ & -   & -   & -   & -   &  -  & $1$ \\
\hline
\end{tabular}
\end{table}
\end{theorem}

\begin{proof}
 Combining Theorem \ref{thhh1}, Lemma \ref{AAA} and Lemma \ref{BBB}, it is
enough to show that any binary systematic AMDS code of parameters $(n,d) \in
\{(6,4), (7,4), (8,4) \}$ is the parity-check code of an AMDS systematic code
of parameters $(n-1,d-1)$. Let $C$ be a binary systematic AMDS code of
parameters $(n,d) \in
\{(6,4), (7,4), (8,4) \}$. Denote by $E$ the code obtained by removing from the
codewords of $C$ the last component. Notice that $E$ is either an MDS code, or
a systematic AMDS code. By Theorem \ref{thhh}, the first case is ruled out.
Hence $E$ is a systematic AMDS code of parameters $(n-1,d-1) \in \{ (5,3),
(6,3), (7,3)\}$ (respectively). Clearly, there exists a function $f:E \to \F_2$
such that $C= \{ (e,f(e) : e \in E) \}$. We will prove that $f$ is the
parity-check funtion on $E$, examining the three cases separately.
\begin{enumerate}
 \item Assume $(n,d)=(6,4)$, so that $E$ has length $n-1=5$ and minimum distance
$d-1=3$. We clearly have $f(0)=0$. By Lemma \ref{AAA}, $E$ has two codewords of
weight three and one of weight four. Let $w \in E$ of weight three. Since $C$
has minimum distance 4, $f(w)=1$. Let $w'$ be the codeword of $E$
of weight 4, and fix $w \in E$ of weight 3. By Lemma \ref{tecn}, we have
$d(w',w) \in \{3,5\}$. If $d(w',w)=5=n-1$, then $w=(1,1,1,1,1)-w'$, which
contradicts $\mbox{wt}(w)=4$. So $d(w',w)=3$. Since $C$ has minimum distance
4 and $f(w)=1$, we must have $f(w')=0$.
\item Assume $(n,d)=(7,4)$, so that $E$ has length $n-1=6$ and minimum distance
$d-1=3$. Again, $f(0)=0$. By Lemma \ref{AAA}, $E$
has four codewords of weight 3 and three of weight 4. Since $C$ has
minimum distance 4, we have $f(w)=1$ for any $w \in E$ of weight 3. Now
fix a codeword $w \in E$ of weight 3 and let $w' \in E$ be any codeword of
weight 4. By Lemma \ref{tecn}, we have $d(w',w) \in \{ 3,5\}$. The case
$d(w',w)=5$ is ruled out by Remark \ref{distwe}. Indeed, Lemma \ref{AAA}
states that $E$ is
linear, and so its distance distribution agrees with its weigh distribution
(given in  Lemma \ref{AAA}). As a consequence, there are no codewords in $E$
whose Hamming distance is five. Since $C$
has minimum distance 4 and $f(w)=1$, we must have $f(w')=0$.
\item Assume $(n,d)=(8,4)$, so that $E$ has length $n-1=7$ and minimum distance
$d-1=3$. We clearly have $f(0)=0$. Since
$d(C)=4$, we get
$f(w)=1$ for any $w \in E$ of weight 3. Let $w' \in
E$ be any
codeword of weight 4 and $w \in E$ a fixed codeword of weight 3. By Lemma
\ref{tecn} and Lemma \ref{BBB},
we have $d(w',w) \in \{3,7\}$. The
case $d(w',w)=7$ is easily ruled out. Since  $f(w)=1$ and
$C$ has minimum distance 4, we have $f(w')=0$. Finally, fix a codeword $w'
\in E$ of weight 4. We have $d(w',\underbrace{1111111}_7)=3$. Since
$f(w')=0$, it is clear that $f(\underbrace{1111111}_7)=1$.
\end{enumerate}

\end{proof}

\begin{remark}
We notice that the $W$-classification of Theorem \ref{thhh2} may be obtained also in the following way. Define an \textbf{isometry} on $\F_q^n$ as a map
$i:\F_q^n \to \F_q^n$ preserving the Hamming distance between elements of $\F_q^n$. Codes $C,D \subseteq \F_q^n$ are said to be \textbf{isometric} if
$D=i(C)$ for some isometry $i$. Combining \cite{O}, Theorem 7.17 and \cite{O},
Table 7.2, we easily see that for any
pair $(n,d) \in \{ (5,3), (6,4), (6,3), (7,4), (8,3), (8,4)\}$
there exists a unique, up to isometry, binary code of length $n$, minimum
distance $d$, and $2^{n-d}$ codewords. 
Since isometric codes have the same distance distribution, we get that 
the $W$-classification of binary systematic AMDS codes must produce a
unique equivalence class for each pair $(n,d)$ in the list. As a
consequence, it is enough to compute the weight distribution of just one code
for each pair
in order to get the whole $W$-classification.
On the other hand, we notice that the proof here proposed uses elementary
techniques, while the classification of \cite{O} refers to 
non-trivial results of design and group theory.
\end{remark}

\section*{Acknowledgment}
The author would like to thank Emanuele Bellini, Elisa Gorla, Simon Litsyn and
Massimiliano Sala
for useful suggestions.


\begin{thebibliography}{99}

% \bibitem{maximal} T. L. Alderson, A. A. Bruen, \emph{Maximal AMDS codes}.
% Applicable Algebra in Engineering, Communication and Computing, vol. 19, issue
% 2, pp. 87-98 (2008).


\bibitem{B} M. A. de Boer, \emph{Almost MDS Codes}. Designs, Codes and
Cryptography 9(2), 1996, pp. 143 -- 155.




\bibitem{pre1} R. D. Baker, J. H. van Lint, R. M. Wilson, \emph{On the
Preparata and Goethals codes}. 
IEEE Transactions on Information Theory, 29, 1983, pp. 342-345.

\bibitem{BGS} E. Bellini, E. Guerrini, M. Sala, \emph{A bound on the size of
linear codes and systematic codes}. \url{http://arxiv.org/abs/1206.6006}.

\bibitem{L} G Cohen, I. Honkala, S. Litsyn, A. Lobstein, \emph{Covering Codes}.
North-Holland Mathematical Library 54, 1997. 


% \bibitem{Dod} S. Dodunekov, I. Landjev, \emph{Near-MDS codes over some small fields}.
% Discrete Mathematics, vol. 213, pp. 55-65 (2000).



% \bibitem{FW} A. Faldum, W. Willems, \emph{Codes of Small Defect}. Designs, Codes
% and Cryptography, 10, pp. 341-350 (1997).

\bibitem{OS} E. Guerrini, M. Sala, \emph{A classification of MDS binary
systematic codes}. BCR preprint, UCC Cork, Ireland, 2006.

% \bibitem{HP} W. C. Huffman, V. Pless, \emph{Fundamentals of Error Correcting
% Codes}. Cambridge University Press, 2003.


\bibitem{O} P. Kaski, P.R.J. \"{O}sterg\aa{}rd, \emph{Classification Algorithms
for Codes and Designs}. Springer-Verlag Berlin Heidelberg, 2006.



% \bibitem{Kim1} D. S. Kim, \emph{Weight Distributions of Hamming Codes}.
% \url{http://arxiv.org/abs/0710.1467}.
% 
% \bibitem{Kim2} D. S. Kim, \emph{Weight Distributions of Hamming Codes (II)}.
% \url{http://arxiv.org/abs/0710.1469}.

\bibitem{LX} S. Ling, C. Xing, \emph{Coding Theory: A First Course}. Cambridge
University Press, 2004.

\bibitem{MS} F. J. MacWilliams, N. J. A. Sloane, \emph{The Theory of
Error-Correcting codes}. North Holland Mathematical Library, 1977. 

\bibitem {NTV} D. Nogin, M. Tsfasman, S. Vl\v{a}du\c{t}, \emph{Algebraic
Geometry Codes, Basic Notions}. American Mathematical Society, Series
\textit{Mathematical Surveys and Monographs}, 2007.


\bibitem{pre} F. Preparata, \emph{A class of optimum nonlinear double-error
correcting
codes}. Information and Control, 13, 1968, pp. 378 -- 400.

\bibitem{Singleton} R. C. Singleton, \emph{Maximum distance $Q$-ary codes}.
IEEE Transactions on Information Theory, 10, 1964, pp. 116 -- 118.

% \bibitem{AA} L. M. G. M. Tolhuizen, \emph{On Maximum Distance Separable codes
% over alphabets of arbitrary size}. 
% IEEE International Symposium on Information Theory, 1994. 


%\bibitem{XMW} X-M. Wang, Y-X. Yang, \emph{On the Undetected Error Probability
%of Nonlinear
%Binary Constant Weight Codes}. IEEE Transactions on Communications, vol. 42,
%n.7 (1994).





\end{thebibliography}
\end{document}